\documentclass[12pt]{iopart}

\usepackage{amsthm,iopams}
\usepackage{setstack}
\usepackage{graphicx}

\newtheorem{thm}{Theorem}[section]
\newtheorem{cor}[thm]{Corollary}
\newtheorem{lem}[thm]{Lemma}

\theoremstyle{definition}

\begin{document}

\title{Dynamical constraints from field line topology in magnetic flux tubes}

\author{A R Yeates, G Hornig}

\address{Division of Mathematics, University of Dundee, Dundee, DD1 4HN, UK}
\eads{\mailto{anthony@maths.dundee.ac.uk}, \mailto{gunnar@maths.dundee.ac.uk}}

\begin{abstract}
A topological constraint on the dynamics of a magnetic field in a flux tube arises from the fixed point indices of its field line mapping. This can explain unexpected behaviour in recent resistive-magnetohydrodynamic simulations of magnetic relaxation. Here we present the theory for a general periodic flux tube, representing, for example, a toroidal confinement device or a solar coronal loop. We show how an ideal dynamics on the side boundary of the tube implies that the sum of indices over all interior fixed points is invariant. This constraint applies to any continuous evolution inside the tube, which may be turbulent and/or dissipative. We also consider the analogous invariants obtained from periodic points (fixed points of the iterated mapping). Although there is a countably infinite family of invariants, we show that they lead to at most two independent dynamical constraints. The second constraint applies only in certain magnetic configurations. Several examples illustrate the theory.
\end{abstract}
\submitto{J. Phys. A: Math. Theor.}

\section{Introduction} \label{sec:intro}

Magnetic flux tubes are basic structural features in astrophysical magnetic fields like the Sun's atmosphere, or may represent the toroidal magnetic field in thermonuclear confinement devices such as tokamaks \cite{goedbloed2004}. Within a single flux tube, the magnetic sub-structure may be extremely complex, particularly when the plasma undergoes a turbulent evolution. A useful technique for characterizing the magnetic field in such a tube is the field line mapping from one end to the other (defined in \sref{sec:notation}). This reduces the field to a two-dimensional function that is simpler to analyze, yet retains much important information. Field line mappings have been used in solar physics, where strong mapping gradients are a likely indicator for preferred locations of magnetic reconnection \cite{titov2002}, and to model laboratory devices with toroidal geometry \cite{morrison2000, borgogno2008}. The latter have exploited the mathematical interpretation of a flux tube as a dynamical system in two space dimensions, with the role of ``time'' played by the direction along the tube. Powerful techniques from dynamical systems theory may thus be applied to study the three-dimensional magnetic structure at a given time snapshot. In this paper, we show how a basic topological property of the field line mapping---the fixed point index---may lead to a physical constraint on the \emph{time evolution} of a magnetic flux tube.

This topological constraint applies under any evolution where field lines on the side boundary of the flux tube (which is by definition a magnetic surface) undergo an ideal evolution that preserves their connectivity \cite{newcomb1958}. Essentially, the interior dynamics are restricted because the total index of interior fixed points is determined by that of fixed points on the boundary (\sref{sec:fixed}). The constraint does not depend on the nature of the dynamics in the tube, which may follow the equations of ideal-magnetohydrodynamics (MHD), resistive MHD, or any other continuous evolution of the magnetic field. However, the work was initially motivated by resistive-MHD simulations of magnetic loops in the solar corona (atmosphere). The coronal magnetic field builds up stress and energy as it is twisted by motion of its footpoints in the solar interior. The key question is: what will be the final state of its subsequent turbulent relaxation? Knowing the appropriate constraints that apply during this relaxation would allow us to predict the final state, and hence place limits on the amount of magnetic energy that can be dissipated during the relaxation. This has direct implications for understanding how the magnetic field heats the corona to extreme temperatures \cite{heyvaerts1984}.

\begin{figure}
\begin{center}
\includegraphics[width=0.4\textwidth]{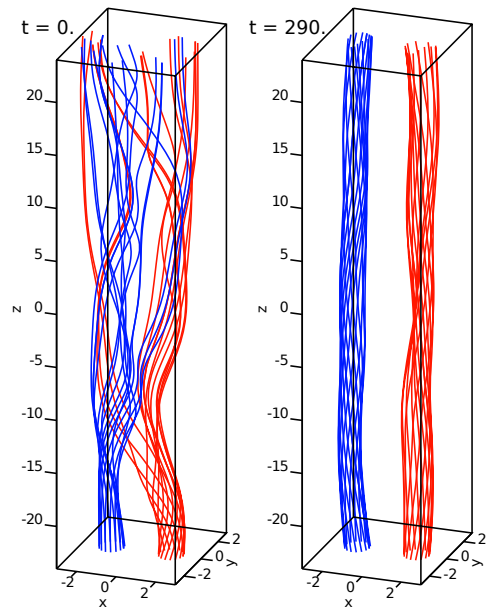}
\caption{Initial and relaxed magnetic fields in the ``braided magnetic field'' simulation (see \cite{wilmotsmith2010,pontin2010} for details and background information). Only selected magnetic field lines are shown (the field fills the whole box), traced from the same starting points on the lower boundary in each case.}
\label{fig:braid}
\end{center}
\end{figure}

The appropriate constraints for a resistive-MHD relaxation remain under debate \cite{bhattacharjee1982}. The theory of Taylor \cite{taylor1974}, which predicts relaxation to a so-called linear force-free field, can explain the final state in certain laboratory experiments, and has been conjectured to apply to the solar corona \cite{heyvaerts1984}. However, its applicability to astrophysical magnetic fields like the corona is unclear, with recent numerical simulations of a solar coronal loop \cite{wilmotsmith2010,pontin2010} finding a final state in conflict with expectations from Taylor theory. \Fref{fig:braid} shows one of these simulations, which we refer to as the ``braided magnetic field'' (due to the initial structure) and use as an illustrative example throughout this paper. We have recently shown that the index constraint can explain why these simulations are unable to reach the predicted Taylor state \cite{yeates2010}.

In this paper, we develop fully the initial idea of \cite{yeates2010} for a general flux tube, explicitly accounting for the (important) influence of the side boundary and exploring the role of higher periodic points. The latter are fixed points of \emph{iterations} of the field line mapping. To allow for these, we focus on a flux tube with periodic boundary conditions (i.e., the distribution of the normal magnetic field is the same across both ends of the flux tube). In that case, the tube end is a Poincar\'{e} section in the dynamical system analogy, and the field line mapping is a return map. Periodicity is natural if the cylinder represents a toroidal geometry, but is also typical in numerical simulations of solar coronal magnetic loops where the initial conditions are given by an analytical equilibrium and the ends remain ``line-tied'' throughout the evolution. In fact, we find that in most cases the higher periodic points do not impose any additional constraint on the dynamics. This result is formulated in Theorem \ref{prop:main} and proved in \sref{sec:periodic}.

\section{The magnetic field line mapping} \label{sec:notation}

\begin{figure}
\begin{center}
\includegraphics[width=0.2\textwidth]{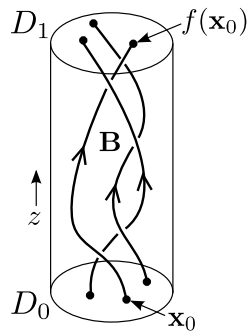}
\caption{The physical setup: a magnetic field ${\bf B}$ on a cylinder of radius $R$.}
\end{center}
\label{fig:tube}
\end{figure}

We shall consider magnetic fields ${\bf B}({\bf x})=B^r{\bf e}_r+B^\phi{\bf e}_\phi + B^z{\bf e}_z$ on the cylinder $\{0\leq r\leq R, 0,\leq \phi \leq 2\pi, 0\leq z \leq 1\}$, sketched in \fref{fig:tube}. The lower and upper boundary discs are denoted $D_0$ and $D_1$ respectively. We impose the following conditions:
\numparts
\begin{eqnarray}
{\bf B} \neq 0,\label{eqn:b0}\\
B^r|_{r=R} = 0, \label{eqn:brR}\\
B^z|_{D_0} = B^z|_{D_1} > 0. \label{eqn:bzs0}
\end{eqnarray}
\endnumparts
In sections \ref{sec:fixed} and \ref{sec:periodic} we will consider time-dependent magnetic fields ${\bf B}({\bf x},t)$, and these conditions will be applied at all times. In particular, \eref{eqn:bzs0} allows the magnetic field on $D_0$ or $D_1$ to be time-dependent, providing that $B_z$ remains periodic. Static ``line-tying'' is a special case.
Since $\nabla\cdot{\bf B}=0$, the magnetic flux is the same through any horizontal cross-section of the cylinder, which therefore represents a magnetic flux tube. Moreover, all magnetic field lines (integral curves of ${\bf B}$) connect $D_0$ and $D_1$ in the same direction, and can be parametrized using the coordinate $z$. The field line starting at ${\bf x}_0\in D_0$ is denoted ${\bf F}({\bf x}_0,z)$ and found by integrating
\begin{equation}
\frac{\partial{\bf F}({\bf x}_0,z)}{\partial z} = \frac{{\bf B}\big({\bf F}({\bf x}_0,z)\big)}{B^z\big({\bf F}({\bf x}_0,z)\big)}
\end{equation}
from $z=0$ to $z=1$, with ${\bf F}({\bf x}_0,0) = {\bf x}_0$. Tracing field lines from every ${\bf x}_0$ yields a mapping $f:D_0\rightarrow D_1$ given by
\begin{equation}
f({\bf x}_0) := {\bf F}({\bf x}_0,1),
\end{equation}
called the \emph{field line mapping} of ${\bf B}$. Since we are interested in physical magnetic fields with finite dissipation, we assume that ${\bf B}$ is differentiable, and hence that $f$ is a diffeomorphism of the disc to itself. It follows from the fact that $f$ is a homeomorphism that the boundary circle $\partial D_0$ is mapped to itself.

From $\nabla\cdot{\bf B}=0$, it may be shown that the Jacobian determinant of $f$ is
\begin{equation}
|Df|({\bf x}_0) = \frac{B^z({\bf x}_0)}{B^z(f({\bf x}_0))},
\label{eqn:df}
\end{equation}
where $Df^{ij}:=\partial f^i/\partial x^j$. Condition \eref{eqn:bzs0} then implies that $|Df|>0$, so the mapping is orientation-preserving. In those fields where $B^z(f({\bf x}_0))=B^z({\bf x}_0)$ for every field line, $f$ is also area-preserving. Such symplectic maps have been widely studied as models for magnetic fields in the fusion context \cite{meiss1992}, but a general field line mapping is only symplectic when written in appropriate canonical coordinates \cite{cary1983}. Although most of our examples do retain this property in physical coordinates, the results in this paper do not depend on it, requiring only conditions \eref{eqn:b0}-\eref{eqn:bzs0}.

\section{Fixed points} \label{sec:fixed}

To analyse the local behaviour of the field line mapping $f$ we introduce, on the disc $D_1$, the mapping
\begin{equation}
{\bf v}_f({\bf x}_0) = f({\bf x}_0) - {\bf x}_0.
\end{equation}
Although ${\bf v}_f$ is not a proper vector field, we can treat it like one for the purposes of this paper. The main use of ${\bf v}_f$ will be to define the fixed point index, but we note first that it leads to a simple but effective ``colour map'' technique for visualising the 2-d mapping $f$ \cite{polymilis2003}. In Cartesian coordinates, each point ${\bf x}_0\in D_0$  is assigned one of four colours, according to the two Cartesian components of ${\bf v}_f$:
\begin{eqnarray*}
\textrm{red if } &v_f^x > 0 &\textrm{ and } v_f^y > 0,\\
\textrm{yellow if } &v_f^x < 0 &\textrm{ and } v_f^y > 0,\\
\textrm{green if }  &v_f^x < 0 &\textrm{ and } v_f^y < 0,\\
\textrm{blue if }  &v_f^x > 0 &\textrm{ and } v^y_f < 0.
\end{eqnarray*}

As a simple example, \fref{fig:elliptic} shows the ``uniform twist'' magnetic field
\begin{equation}
{\bf B} = {\bf B}^{\rm twist} + {\bf e}_z,\qquad
{\bf B}^{\rm twist}=-\alpha y{\bf e}_x + \alpha x{\bf e}_y,
\label{eqn:unitwist}
\end{equation}
for two values of the twist-angle parameter $\alpha$. The arrows in \fref{fig:elliptic} show the direction of ${\bf v}_f$, which for this field may be shown to be
\numparts
\begin{eqnarray}
v_f^x= (\cos\alpha) x_0 -(\sin\alpha) y_0 -x_0,\\
v_f^y=(\sin\alpha) x_0 + (\cos\alpha) y_0 - y_0.
\end{eqnarray}
\endnumparts

\begin{figure}
\begin{center}
\includegraphics[width=0.55\textwidth]{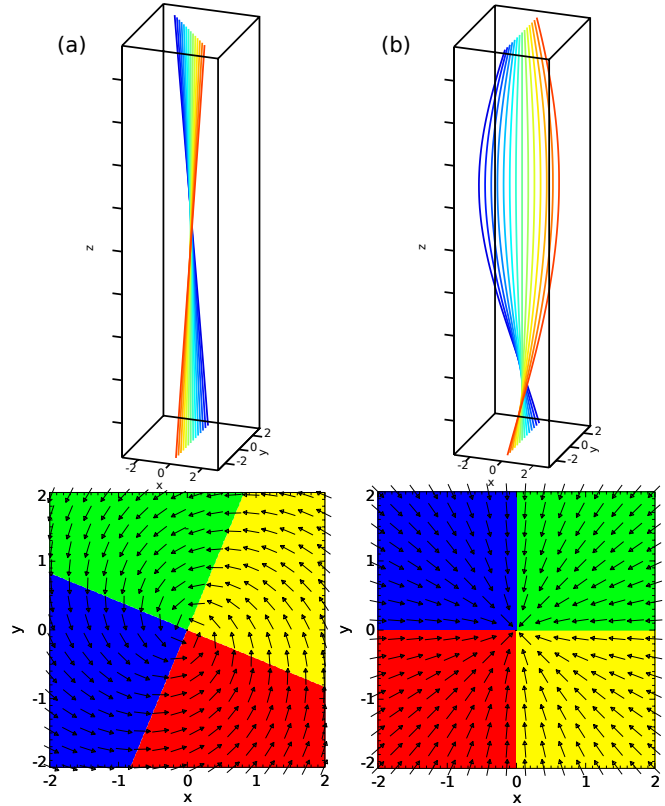}
\caption{The uniform twist field \eref{eqn:unitwist} with (a) $\alpha=\pi/4$ and (b) $\alpha=\pi$. The vectors show the direction of ${\bf v}_f$, normalised to unit length for clarity. Here $R=2$.}
\label{fig:elliptic}
\end{center}
\end{figure}

A more complex example is given by the braided magnetic field simulation introduced in \sref{sec:intro}. A sequence of colour maps obtained at different times during this simulation is shown in  \fref{fig:braid_cmaps}. It neatly reveals how the field line mapping is simplified during the relaxation.

\begin{figure}
\begin{center}
\includegraphics[width=0.85\textwidth]{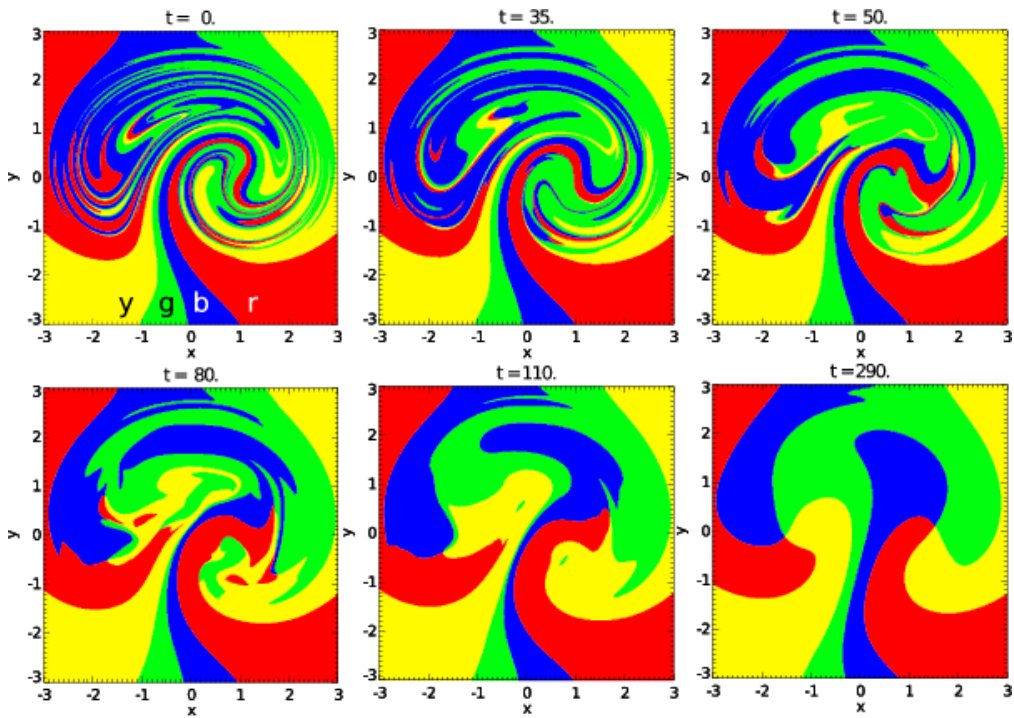}
\caption{Time sequence of colour maps from the braided magnetic field simulation (see \cite{pontin2010} for details of the simulation). Letters y, g, b, r indicate the colours yellow, green, blue, and red (for readers viewing in greyscale).}
\label{fig:braid_cmaps}
\end{center}
\end{figure}

Central to our analysis are the \emph{fixed points} where $f({\bf x}_0)={\bf x}_0$, or equivalently ${\bf v}_f=0$. Isolated fixed points show in the colour map as intersections of the curves $v_f^x=0$ and $v_f^y=0$, i.e., points where all four colours meet.  Indeed, the colour map was introduced to aid the numerical search for fixed points \cite{polymilis2003}. The uniform twist field \eref{eqn:unitwist} has a single fixed point at the origin $x_0=y_0=0$, while the braided magnetic field has multiple fixed points at different times.

The reason for focusing on fixed points of $f$ is simple: they are \emph{topological} properties of the mapping, and hence of ${\bf B}$. In an ideal time evolution, the fixed points are invariants. But even under a non-ideal evolution they cannot be arbitrarily created or destroyed. Importantly, we can make a stronger statement not just about the fixed points in isolation but about the local structure of the mapping around them. This is expressed through a topological property called the (Poincar\'{e}/-Hopf) \emph{index} of each fixed point.

\subsection{Index of a fixed point}

The index of a fixed point of $f$ depends on the local mapping around the fixed point, so to assign an index to fixed points lying on the boundary $\partial D_0$, the mapping $f$ must be extended outside the disc $r=R$. To this end, following \cite{brown1994}, define the extended mapping
\begin{equation}
\tilde{f}({\bf x}_0)  = 
\cases{f({\bf x}_0) & \textrm{ if $|{\bf x}_0|\leq R$,}\\
f\big(R{\bf x}_0/|{\bf x}_0|\big) & \textrm{ if $|{\bf x}_0| > R$,}\\}
\label{eqn:ftilde}
\end{equation}
which is continuous and introduces no additional fixed points. (An alternative extension is used in \cite{ma2001}.)

Let $\Gamma$ be a closed curve enclosing an isolated fixed point ${\bf x}_0\in D_0$ but enclosing no other fixed point. The \emph{index} of ${\bf x}_0$, denoted $\textrm{ind}_{{\bf x}_0}f$, is the winding number of the closed curve ${\bf v}_{\tilde{f}}(\Gamma)$ about the origin \cite{brown1994}. It is an integer and so invariant under homotopy: it is a \emph{topological} property of the local mapping around the fixed point. The extension $\tilde{f}$ means that this definition applies both to interior fixed points and to those on $\partial D_0$. Note that there are a number of equivalent definitions of the index \cite{simon1974}. The winding number may be expressed in integral form as
\begin{equation}
\textrm{ind}_{{\bf x}_0}f = \frac{1}{2\pi}\oint_\Gamma \rmd\left[\textrm{arctan}\left(\frac{v_{\tilde{f}}^y}{v_{\tilde{f}}^x}\right)\right] = \frac{1}{2\pi}\oint_\Gamma\frac{v_{\tilde{f}}^x\,\rmd v_{\tilde{f}}^y - v_{\tilde{f}}^y\,\rmd v_{\tilde{f}}^x}{v_{\tilde{f}}^2},
\label{eqn:kronecker}
\end{equation}
which is the two-dimensional case of the Kronecker integral \cite{polymilis2003}. Evaluated around a closed curve enclosing more than one fixed point, the integral yields the algebraic sum of their indices.

The fixed point in the uniform twist example has index $+1$. As another simple example, \fref{fig:hyp} shows the ``hyperbolic'' magnetic field
\numparts
\begin{eqnarray}
\fl {\bf B} = {\bf B}^{\rm hyp}(r,\phi) + {\bf e}_z, \label{eqn:hyp}\\
\fl {\bf B}^{\rm hyp}(r,\phi) = 2\cos(2\phi)\sin\left(\frac{\pi r}{2}\right){\bf e}_r - \left[ \sin\left(\frac{\pi r}{2}\right) + \frac{\pi r}{2}\cos\left(\frac{\pi r}{2}\right)\right]\sin(2\phi){\bf e}_\phi,
\end{eqnarray}
\endnumparts
comprising a uniform vertical component $B^z=1$ and horizontal components given by ${\bf B}^{\rm hyp}$. The extension of the colour map outside the disc $R=2$ (using $\tilde{f}$) is also shown. Here $f$ has five interior fixed points and four on the boundary. The fixed point at $x_0=y_0=0$ has index $-1$ (hence the name ``hyperbolic''), while the remaining four interior fixed points all have index $+1$. Of the boundary fixed points, the pair at $x_0=0$ have index $0$, while those at $y_0=0$ have index $-1$.

\begin{figure}
\begin{center}
\includegraphics[width=0.42\textwidth]{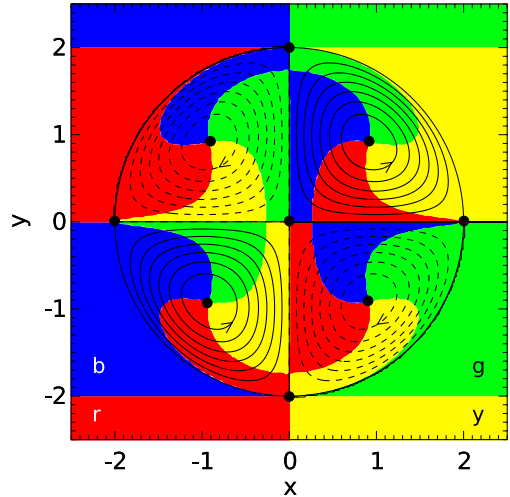}
\caption{Colour map of the ``hyperbolic'' magnetic field \eref{eqn:hyp}, with projected field lines showing the horizontal component ${\bf B}^{\rm hyp}$ (not the vector field ${\bf v}_f$). Black circles identify the nine fixed points of the field line mapping.}
\label{fig:hyp}
\end{center}
\end{figure}

\subsection{Global constraint}

The utility of the fixed point index is encapsulated in the following global result connecting the indices of all of the fixed points of $f$ \cite{hopf1929,brown1994}. We denote the set of fixed points of $f$ by $\textrm{Fix}(f)$.

\begin{thm}[Hopf]
If $M$ is the closure of an open, connected subset of the plane whose boundary is a union of smooth simple closed curves and $f,g:M\rightarrow M$ are homotopic maps with both $\textrm{Fix}(f)$ and $\textrm{Fix}(g)$ finite, then
\[
\sum_{{\bf p}\in\textrm{Fix}(f)}\textrm{ind}_{\bf p}f = \sum_{{\bf p}\in\textrm{Fix}(g)}\textrm{ind}_{\bf p}g.
\]
\label{thm:hopf}
\end{thm} 

In our case, $M$ is the disc $D_0$, and the following simple argument shows that the value of the index sum for any $f$ is unity \cite{brown1994}. Since the disc is contractible, all self-maps are homotopic. Thus any map $f$ is homotopic to the ``constant'' map $c:D\rightarrow D$ where $c({\bf x}_0)=0$ for all ${\bf x}_0\in D$. This map has a single fixed point at ${\bf x}_0=0$ with index $1$. Hence, by Theorem \ref{thm:hopf}, the map $f$ must also have index sum $1$, provided that it has a finite set of fixed points. For our further analysis it will be useful to separate the sum of interior fixed points from those on the boundary. We define
\begin{equation}
T_{\textrm{int}} = \sum_{{\bf p}\,\in\,\textrm{Fix}(f)\cap \textrm{int}(D_0)} \textrm{ind}_{\bf p} f, \qquad
T_{{\partial}} = \sum_{{\bf p}\,\in\,\textrm{Fix}(f)\cap\partial D_0} \textrm{ind}_{\bf p} f,
\end{equation}
giving us the following result.
\begin{cor}
If $f$ is a continuous self-map of the disc $D_0$ with $\textrm{Fix}(f)$ finite, then
\begin{equation}
T_{\textrm{int}} + T_{\partial} = 1.
\label{eqn:indrel}
\end{equation}
\label{cor:sum}
\end{cor}
In the absence of boundary fixed points, the interior sum $T_{\textrm{int}}$ is known as the \emph{Lefschetz number}, or \emph{topological degree} of the mapping $f$. For the hyperbolic field of \fref{fig:hyp}, it is readily determined that $T_{\textrm{int}}=3$ and $T_{\partial}=-2$, verifying Corollary \ref{cor:sum} for this particular example.

The relation \eref{eqn:indrel} will be satisfied by the field line mapping of any continuous magnetic field ${\bf B}$, at all times during its dynamical evolution. In general, a magnetic field will be able to exchange contributions between the two terms $T_{\textrm{int}}$ and $T_{\partial}$, in such a way that their sum remains unity. However, if the physical situation is such that $T_{\partial}$ remains fixed, then Corollary \ref{cor:sum} implies that $T_{\textrm{int}}$ must also be a conserved quantity in the dynamical evolution, irrespective of the details of the dynamics. This is not an unusual situation: for instance, in astrophysical plasmas---where the evolution is commonly ideal in most of the domain---it is often possible to choose a flux tube such that the boundary remains in the ideal region while the interesting (non-ideal) dynamics occur in the interior.

The braided magnetic field simulation is a case where this topological constraint is significant and determines the final state of the system \cite{yeates2010}. The sequence of colour maps in \fref{fig:braid_cmaps} shows how, although the number of fixed points decreases from 27 at $t=0$ to 2 at $t=290$ (through bifurcations which preserve the total index), the overall value $T_{\textrm{int}}=2$ is conserved. Note that, in this simulation, the evolution at the side boundary of the computational box remains ideal, so that the field line mapping is preserved there. \Fref{fig:braid_cmaps} shows only part of the domain between $x=\pm3, y=\pm3$, but this region contains all of the interior fixed points and all of the non-ideal dynamics. The total index $T_{\textrm{int}}$ may be determined either by calculating the Kronecker integral \eref{eqn:kronecker} around the boundary of this square region \cite{polymilis2003}, or simply by inspecting the sequence of colours around this boundary, which suffice to determine the winding number of ${\bf v}_{\tilde{f}}$. The latter is clearly seen to remain the same throughout the simulation.

\subsection{Resistive diffusion: an apparent paradox}

For a fixed normal flux distribution on the boundary of the domain, it is well known that under resistive diffusion in a static medium, where
\begin{equation}
\frac{\partial {\bf B}}{\partial t} = \eta\nabla^2 {\bf B},
\end{equation}
every initial field ${\bf B}$ must decay asymptotically to a unique minimum-energy potential field (satisfying $\nabla\times{\bf B}=0$). But, for a given normal flux distribution on the boundary, one can have fields with a wide variety of $T_{\rm int}$. How can $T_{\rm int}$ be conserved if it differs from $T_{\rm int}$ of the minimum-energy field? 

The resolution of this paradox is suggested by the braiding simulation. Here the dynamical relaxation leads to a final state that preserves $T_{\rm int}=2$, even though the minimum-energy field is ${\bf B}_{\rm pot}={\bf e}_z$, a uniform vertical field where the index is undefined (since all field lines are periodic). If the simulation is continued further, the horizontal magnetic field continues to decay, but only on the much slower global resistive timescale $\tau = L^2/\eta$, where $L$ is the global domain size. This decay suggests that the vertical field ${\bf B}_{\rm pot}$ will be reached, but only asymptotically as $t\rightarrow\infty$. This infinite limit is usually not of physical interest; the question is how the system evolves on the much shorter dynamical timescale. At all \emph{finite} times, while any horizontal magnetic field remains, we expect $T_{\rm int}$ to be preserved (providing that $T_\partial$ is preserved). Further investigation of this point is merited in more general configurations.

\subsection{Generic fixed points}

What are the possible values of $\textrm{ind}_{{\bf x}_0}f$ for a fixed point? Non-degenerate interior fixed points correspond to non-degenerate critical points of ${\bf v}_f$, so have index either $+1$ or $-1$. Non-degenerate means that $\det(Df - Id)\neq 0$. A degenerate fixed point can always be decomposed by a small perturbation of the mapping into several non-degenerate fixed points, such that their total index equals that of the degenerate fixed point. Such degenerate fixed points are thus topologically unstable. In our physical system, we assume that all interior fixed points will be non-degenerate, or ``generic''.

\begin{figure}
\begin{center}
\includegraphics[width=0.4\textwidth]{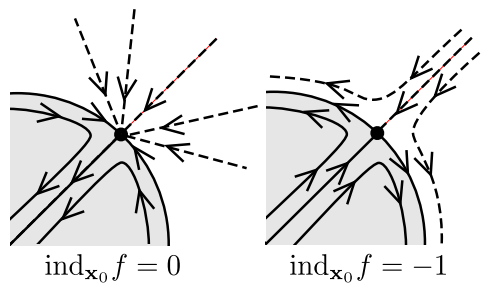}
\caption{The local vector field ${\bf v}_{\tilde{f}}$ for the two generic types of boundary fixed point, in the case that $|Df|=1$ at the fixed point. Dashed lines indicate the extension outside $D_0$, which is determined completely by the mapping on $\partial D_0$.} 
\label{fig:boundaryfp}
\end{center}
\end{figure}

The boundary condition \eref{eqn:brR} restricts the possible types of fixed point that we can have on the boundary $\partial D_0$.
\begin{lem}
Let ${\bf x}_0$ be an isolated, non-degenerate fixed point of $f$ located on the boundary $\partial D_0$. If $\left|Df \right|({\bf x}_0)=1$ then $\textrm{ind}_{{\bf x}_0}f$ is either $0$ or $-1$.
\label{lem:bound}
\end{lem}
The proof is given in \ref{app:proof}, while the two possible cases are shown in \fref{fig:boundaryfp}. Observe that our periodic condition \eref{eqn:bzs0}, along with \eref{eqn:df}, ensures $\left|Df \right|({\bf x}_0)=1$ at every fixed point ${\bf x}_0$. Thus Lemma \ref{lem:bound} applies to every fixed point on $\partial D_0$, implying that $T_{\partial} \leq 0$ for any $f$. So, by Corollary \ref{cor:sum}, a periodic flux tube must have
\begin{equation}
T_{\textrm{int}} \geq 1.
\end{equation}

\section{Periodic points} \label{sec:periodic}

The success of Corollary \ref{cor:sum} in explaining the final state of the braided magnetic field \cite{yeates2010} leads us to ask whether additional topological constraints arise from higher periodic points. This requires condition \eref{eqn:bzs0} that the flux tube is periodic.

Writing $f^m$ for $f\circ f\circ \ldots \circ f$ iterated $m$ times, a \emph{periodic point} ${\bf x}_0\in D_0$ of \emph{(minimum) period} $m$ satisfies
\numparts
\begin{eqnarray}
f^m({\bf x}_0) = {\bf x}_0,\\
f^n({\bf x}_0) \neq {\bf x}_0 \,\textrm{ for $n<m$}.
\end{eqnarray}
\endnumparts
A periodic point is thus a fixed point of $f^m$. In particular, the fixed points of $f$ are periodic points with period 1. \Fref{fig:period2} is a sketch of the second iteration, $f^2$, for a hypothetical flux tube. Notice that the fixed points of $f^2$ include periodic points of periods 1 and 2. In general, the fixed points of a given iteration $f^q$ could include periodic points of any period $m\geq 1$ that divides $q$ (written $m|q$). Notice also that the periodic points of period $m$ fall into groups of size $m$ corresponding to the same ``periodic orbit''. In \fref{fig:period2}, the blue periodic orbit comprises two period 2 points, ${\bf x}_2^a$ and ${\bf x}_2^b$.

\begin{figure}
\begin{center}
\includegraphics[width=0.2\textwidth]{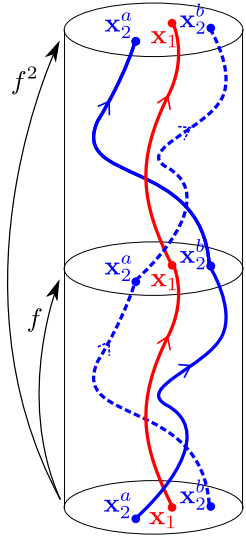}
\caption{Field lines corresponding to some fixed points of $f^2$. The point ${\bf x}_1$ (in red) is a period 1 point (i.e. a fixed point of $f$). The points ${\bf x}_2^a$, ${\bf x}_2^b$ (in blue) are period 2 points, both belonging to the same periodic orbit.}
\label{fig:period2}
\end{center}
\end{figure}

For a given iteration $q$ we can consider (as before) the sum over all interior fixed points,
\begin{equation}
T_{\textrm{int}}^q = \sum_{{\bf p}\,\in\,\textrm{Fix}(f^q)\cap\textrm{int}(D_0)} \textrm{ind}_{\bf p} f^q,
\end{equation}
and the analogous sum $T_{\partial}^q$ over the boundary fixed points.

\begin{cor}
If $f$ is a continuous self-map of the unit disc $D_0$, with $\textrm{Fix}(f^q)$ finite for every $q\in\mathbb{N}$, then
\begin{equation}
T_{\textrm{int}}^q + T_{\partial}^q = 1 \textrm{ for all $q\in\mathbb{N}$}.
\end{equation}
\label{cor:family}
\end{cor}

If all $T_{\partial}^q$ are conserved during the dynamical evolution, then the $T_{\textrm{int}}^q$ will be a countably infinite family of topological constraints on the magnetic field evolution. The question of interest physically is whether these additional constraints are important in practice, or, in other words, whether they are independent of one another. The following result shows that, in a large class of magnetic fields, they are not.

\begin{thm}
Let $f$ be an orientation-preserving homeomorphism of the unit disc $D_0$, with $\textrm{Fix}(f^q)$ finite for every $q\in\mathbb{N}$, and with $\textrm{ind}_{{\bf x}_0}f^q\in\{-1,0\}$ for every periodic point ${\bf x}_0$ on $\partial D_0$. Then
\begin{enumerate}
\item If $T_{\textrm{int}}^1 > 1$ then $T_{\textrm{int}}^q = T_{\textrm{int}}^1$ for all $q\in\mathbb{N}$.
\item If $T_{\textrm{int}}^1 = 1$ then $T_{\textrm{int}}^q = 1\,\forall q\in\mathbb{N}$ if and only if the rotation number of $f|_{\partial D_0}$ is irrational.
\end{enumerate}
\label{prop:main}
\end{thm}

According to Theorem \ref{prop:main}, the index sums of higher iterations can differ from that of the initial mapping only if the initial sum is $1$, and then only if the restriction of $f$ to the boundary $\partial D_0$ has a particular form (namely a rational rotation number: this will be defined below). Thus for any other initial magnetic field, it will hold at all times throughout the ensuing dynamical evolution that $T_{\textrm{int}}^q = T_{\textrm{int}}^1$ for all $q\in\mathbb{N}$. This means that the higher constraints of Corollary \ref{cor:family} will be automatically maintained for any possible dynamics and thus pose no additional constraint on the evolution beyond that imposed by $T_{\textrm{int}}^1$.

\subsection{Proof of Theorem \ref{prop:main}}

The strategy is to use the fact that $T_{\textrm{int}}^q=1-T_{\partial}^q$ and focus on the boundary mapping $f|_{\partial D_0}$. Firstly, we consider the sequence of indices of an individual fixed point under iteration of $f$. 

A fixed point ${\bf x}_0$ of $f$ recurs as a fixed point of each iteration $f^q$ for $q\in\mathbb{N}$, and the index of this fixed point under each $f^q$ generates a sequence $\big\{\textrm{ind}_{{\bf x}_0}f^q\big\}$. The possible index sequences for different classes of mapping have been quite well studied (see, e.g., \cite{graff2005}). If $f$ is continuous and ${\bf x}_0$ is an isolated fixed point for each $q$, then Dold \cite{dold1983} proved that the sequence of indices satisfies the sequence of congruences
\begin{equation}
\sum_{n|q}\mu(n)\,\textrm{ind}_{{\bf x}_0}f^{q/n} \equiv 0 \quad\textrm{mod } q
\label{eqn:dold}
\end{equation}
for each $q\geq 1$ (called \emph{Dold relations}). Here $\mu(n)$ is the M\"{o}bius function from number theory, defined by
\begin{equation}
\fl\mu(n)=
\cases{0 & \textrm{if $p^2|n$ for some prime $p$},\\
1 & \textrm{if $n$ is square-free with an even number of prime factors},\\
-1 & \textrm{if $n$ is square-free with an odd number of prime factors}.}
\end{equation}

Since the extended mapping $\tilde{f}$ is continuous, we may use the Dold relations \eref{eqn:dold} to obtain the following Lemma for fixed points on $\partial D_0$. Note that it follows from \eref{eqn:ftilde} that the fixed points of $\tilde{f^q}$ are the same as those of $f^q$ for all $q$.

\begin{lem}
Let $f$ be a field line mapping whose boundary fixed points on each iteration have index either $0$ or $-1$. Then for each such fixed point ${\bf x}_0$,
\[
\textrm{ind}_{{\bf x}_0}f^q = \textrm{ind}_{{\bf x}_0}f \textrm{ for all $q\in\mathbb{N}$}.
\]
\label{lemma}
\end{lem}

\begin{proof}
Induction on $q$ (cf. a similar argument used by \cite{graff2003}). Clearly the result holds for $q=1$. Suppose that $\textrm{ind}_{{\bf x}_0}f^m = \textrm{ind}_{{\bf x}_0}f$ for all $m<q$. Now, separating the first term in the $q$th Dold's relation \eref{eqn:dold} gives
\[
\mu(1)\,\textrm{ind}_{{\bf x}_0}f^q + \sum_{\substack{n|q\cr n>1}}\mu(n)\,\textrm{ind}_{{\bf x}_0}f^{q/n} \equiv 0 \quad\textrm{mod } q.
\]
Now $\mu(1)=1$, and since $q/n < q$ the inductive assumption gives
\[
\,\textrm{ind}_{{\bf x}_0}f^q + \textrm{ind}_{{\bf x}_0}f\sum_{\substack{n|q\cr n>1}}\mu(n) \equiv 0 \quad\textrm{mod } q.
\]
But
\[
\sum_{\substack{n|q\cr n>1}}\mu(n) = \sum_{n|q}\mu(n) - \mu(1) = -1,
\]
using a well-known property of the M\"{o}bius function. Thus we must have
\[
\textrm{ind}_{{\bf x}_0}f^q \equiv \textrm{ind}_{{\bf x}_0}f \quad\textrm{mod } q,
\]
and the result follows from the restriction that $\textrm{ind}_{{\bf x}_0}f^q\in\{-1,0\}$.
\end{proof}

Now consider the restriction $f|_{\partial D_0}$ of $f$ to the boundary. This is an orientation-preserving homeomorphism of the circle. (Orientation-preserving here means that the order of any three points on $\partial D_0$ cannot be permuted by $f$. Physically, it follows from the fact that $f$ maps $\partial D_0$ to itself, and the restriction that magnetic field lines cannot intersect when ${\bf B}\neq 0$.) But an orientation-preserving circle homeomorphism can have periodic points of at most a single (minimal) period \cite{demelo1993}. Thus there are two cases.

\emph{Case (i)}: $T_{\textrm{int}}^1 > 1$. In this case we must have $T_{\partial}^1\neq 0$, so that $f$ has at least one fixed point on the boundary. But then $f$ can have no boundary periodic points of any higher period: the fixed points of $f^q$ will be the same as those of $f$ for all $q\in\mathbb{N}$. Applying Lemma \ref{lemma} to each boundary fixed point gives $T_{\textrm{int}}^q=T_{\textrm{int}}^1$ for all $q$.

\emph{Case (ii)}: $T_{\textrm{int}}^1 = 1$. This implies that $f$ has no fixed points on $\partial D_0$. Therefore there could be periodic points of some unique period $m>1$. According to the theory of circle homeomorphisms, the existence of such points depends on the rotation number of the homeomorphism \cite{demelo1993}. The \emph{rotation number} of $f|_{\partial D_0}$ is defined as
\[
\rho = \lim_{k\rightarrow\infty}\frac{G^k(\phi_0)-\phi_0}{2\pi k},
\]
where $\phi_0\in\partial D_0$ and $G$ is a lift of $f|_{\partial D_0}$ to the real line. The limit exists and is independent of the point $\phi_0$. Since $f|_{\partial D_0}$ is an orientation-preserving circle homeomorphism, $\rho$ is rational if and only if $f|_{\partial D_0}$ has at least one periodic point. In that case, and if $\rho$ written in its lowest terms is $n/m$, then all periodic points must have period $m$. Otherwise, $\rho$ is irrational and $f|_{\partial D_0}$ has no periodic points. We now apply Lemma \ref{lemma}. In the irrational case, we have $T_{\partial}^q=0$ for all $q\in\mathbb{N}$. In the rational case, we have
\begin{equation}
T_\partial^q=
\cases{T_m &\textrm{ if $m|q$},\\
0 & \textrm{ otherwise},}
\end{equation}
where $T_m<0$ is a constant given by the index sum of the $m$th iteration (the first with fixed points). This completes the proof of Theorem \ref{prop:main}. \qed

\subsection{Examples} \label{sec:eg}

The following examples illustrate the two possible cases in Theorem \ref{prop:main}. Case (i) where $T_{\textrm{int}}^1>1$ is exemplified by the braided magnetic field simulation discussed earlier. \Fref{fig:pbraid} shows computed colour maps for $f$, $f^2$, and $f^3$ in both the initial and final states of this simulation. At both times, $T^1_{\textrm{int}}=2$, and we see that $T_{\textrm{int}}^2=T_{\textrm{int}}^3=2$, in accordance with Theorem \ref{prop:main}. Thus the higher $T_{\textrm{int}}^q$ are automatically preserved in this simulation because $T_{\textrm{int}}^1$ is preserved and the dynamics are ideal near to $\partial D_0$. Notice that, although the index sums are equal for each iteration, this does not preclude the appearance of more periodic points at each iteration. This is visible particularly in the final state, where there are $2$ fixed (period 1) points, $16$ period 2 points, and $36$ period 3 points. However, the additional periodic points appear in equal numbers with index $+1$ and $-1$.

\begin{figure}
\begin{center}
\includegraphics[width=0.85\textwidth]{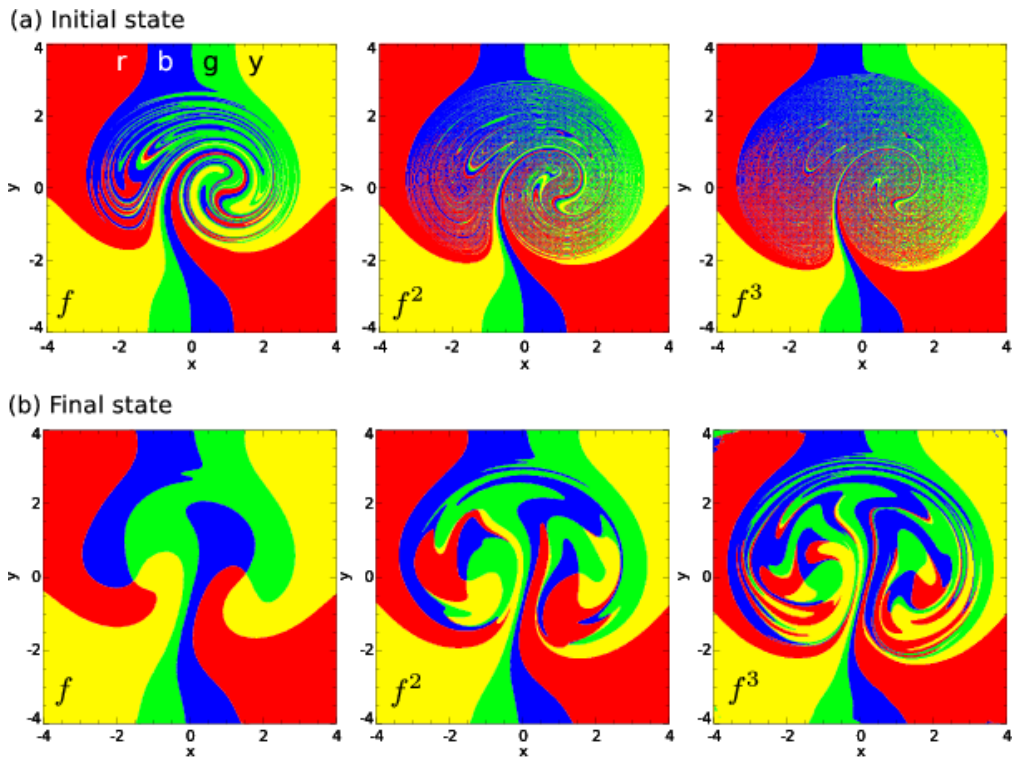}
\caption{Colour maps for the first three iterations of $f$ in (a) the initial state and (b) the final state of the original braided magnetic field simulation.}
\label{fig:pbraid}
\end{center}
\end{figure}

To illustrate case (ii) in Theorem \ref{prop:main}, we begin by constructing a magnetic field whose field line mapping has $T_{\rm int}^1=1$ but $T_{\rm int}^2\neq1$. This may be done by taking the uniform twist field ${\bf B}^{\rm twist} + {\bf e}_z$ and adding a small perturbation. We choose $\alpha=\pi$ so that every point of $D$ is a period 2 point of ${\bf B}^{\rm twist} + {\bf e}_z$. The perturbation takes the form of ${\bf B}^{\rm hyp}$, except that it is rotated as $z$ increases at the same rate as the uniform twist. This destroys the symmetry of the uniform twist map, leaving a finite set of period 2 points, which correspond to the fixed points of the original ${\bf B}^{\rm hyp}$ field. Our overall field is
\begin{equation}
{\bf B} = {\bf B}^{\rm twist} + \epsilon{\bf B}^{\rm hyp}(r,\phi-\pi z) + {\bf e}_z,
\label{eqn:combo}
\end{equation}
where $\epsilon=0.1$ and $\alpha=\pi$. The boundary mapping has rotation number $\rho=1/2$, so has periodic points of minimal period 2 only. This is verified for the first four iterations of $f$ by the colour maps shown in \fref{fig:comb}. In particular, we find $T_{\rm int}^2 = 3$, and so for this mapping
\begin{equation}
T_{\textrm{int}}^q =\cases{1 &\textrm{ for $q$ odd,}\\
3 &\textrm{ for $q$ even.}}
\end{equation}

\begin{figure}
\begin{center}
\includegraphics[width=\textwidth]{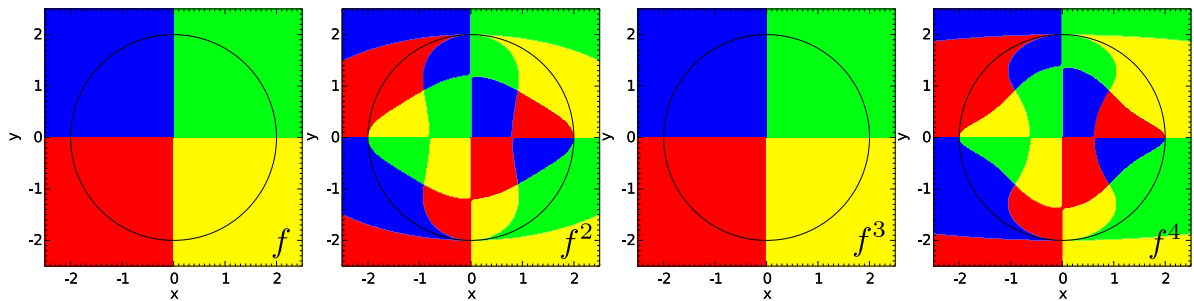}
\caption{Sequence of colour maps for iterations of the perturbed uniform twist magnetic field \eref{eqn:combo}, where all four periodic points on the boundary have period 2.}
\label{fig:comb}
\end{center}
\end{figure}

If a magnetic field with this boundary mapping were to evolve subject to ideal dynamics at the boundary, then our theory predicts two constraints on the evolution. We have verified this prediction through a new numerical simulation, solving the same resistive-MHD equations as the earlier braiding simulation \cite{wilmotsmith2010,pontin2010}. To initiate a dynamical relaxation, we add to the background field \eref{eqn:combo} six twist regions. These twist regions match those in the initial condition of the original braided magnetic field simulation (see \cite{wilmotsmith2009} for details). Their effect is to add significant complexity to the field line mapping, triggering the formation of thin current sheets in the relaxation process. (Since the initial magnetic field is smooth, it must remain smooth for any finite time $t$ during the resistive evolution, so the current sheets have a finite width rather than being true discontinuities.)

The top row of \fref{fig:combosim} shows $f$ and $f^2$ for the initial condition of the new simulation. For computational reasons, we use a rectangular domain $-6\leq x\leq 6$, $-6\leq y \leq 6$ and $-24 \leq z \leq 24$. The initial background field \eref{eqn:combo} has been scaled in $r$ and $z$ so that the cylindrical flux tube boundary is located at $R=5.22$ (shown by the circle) and the overall twist remains $\pi$. Because the superimposed twist regions are localised well within the flux tube, the field line mapping at $r=R$ follows that of the background field (a uniform twist of $\pi$ with a small perturbation). Hence the overall initial field for $r<R$ retains the property of \eref{eqn:combo} that $T_{\rm int}^1=0$ and $T_{\rm int}^2 = 3$. The bottom row of \fref{fig:combosim} shows that this topological property persists throughout the dynamical evolution, verifying that both $T_{\rm int}^1$ and $T_{\rm int}^2$ are conserved, acting as dynamical constraints. The simulation is imperfect in that the cylindrical flux tube boundary does not remain perfectly ideal, due to the finite numerical dissipation and to the fact that it does not coincide with the computational boundary. Thus the period 2 points on the flux tube boundary move slightly between the initial and final states. Nevertheless, it remains sufficiently close to ideal to illustrate the topological constraints.

\begin{figure}
\begin{center}
\includegraphics[width=0.65\textwidth]{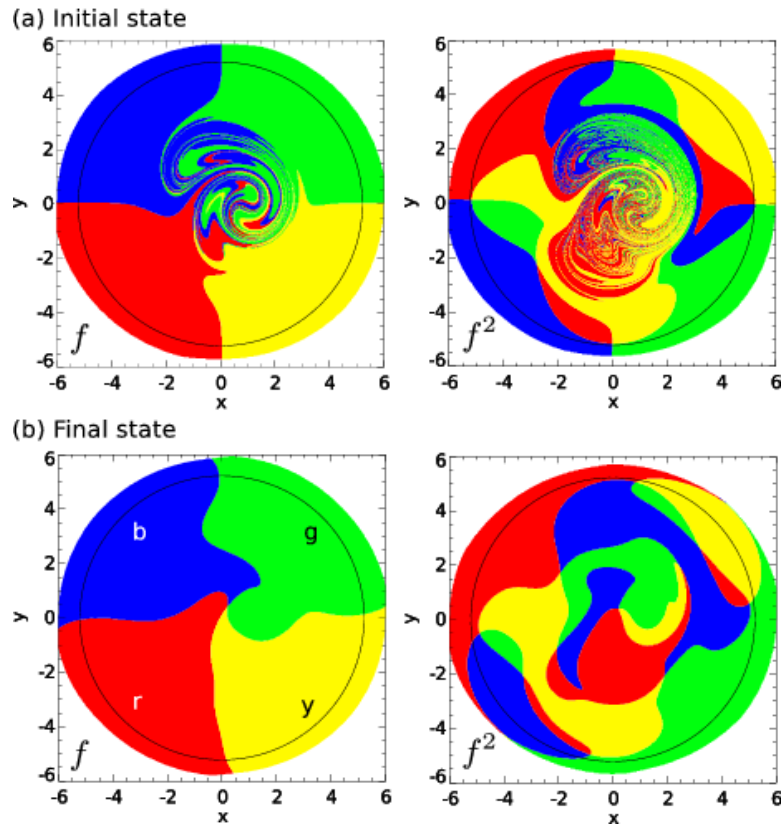}
\caption{Colour maps for the first two iterations of $f$ in (a) the initial state and (b) the final state of the new simulation described in \sref{sec:eg}. This is a magnetic field with two independent dynamical constraints. Outside the flux tube $R=5.22$ (shown by the black circle), white points indicate field lines which leave the side boundary of the computational box (where they are ``line-tied'').}
\label{fig:combosim}
\end{center}
\end{figure}

\ack
We thank G. Graff for useful suggestions on an initial draft, and two anonymous referees for helpful comments. The work was supported by STFC Grant ST/G002436/1. Numerical simulations used the Lare3d code \cite{arber2001} and the UKMHD consortium parallel cluster at the University of St Andrews funded by STFC and SFC (SRIF).

\appendix
\section{Proof of Lemma \ref{lem:bound}} \label{app:proof}

This relies on the boundary condition \eref{eqn:brR}, and on the property that $f$ maps the boundary $\partial D_0$ to itself. Without loss of generality we work in Cartesian coordinates $(x,y)$ where the $x$-axis corresponds to  $\partial D_0$ and the upper half of the $xy$-plane to the interior of $D_0$. The fixed point is at $x=y=0$. Let ${\bf V}_f$ be the local linearisation of ${\bf v}_f$ about the fixed point, so that ${\bf V}_f=A{\bf x}$, where the matrix $A$ is given by
\begin{equation}
A \equiv \left(
\begin{array}{cc}
a & b \\
c & d
\end{array}
\right) = Df(0,0) - Id.
\label{eqn:amat}
\end{equation}

We start by noting that, since $f$ maps the boundary $y=0$ to itself, $(1,0)^T$ must be an eigenvector of $A$. This implies that $a=\lambda_1$ and $c=0$, where $\lambda_1$ is the eigenvalue corresponding to the direction along the boundary, and that both eigenvalues are real. It follows from the relation $\Tr A=\lambda_1+\lambda_2$ that $d=\lambda_2$.

Next, we derive a constraint on $\lambda_2$ from the boundary condition \eref{eqn:brR}. This condition may be expressed as $f^y(x,y) > 0$ for $y>0$. Since $f^y(x,0)=0$, for small $y$ we may write
\begin{equation}
f^y \approx \left.\frac{\partial f^y}{\partial y}\right|_{(0,0)}y.
\end{equation}
Hence we must have $\partial f^y/\partial y|_{(0,0)}>0$. Forming the product of both sides of \eref{eqn:amat} with ${\bf x}=(x,y)^T$ then gives the following equality for the $y$-component:
\begin{equation}
\left(\left.\frac{\partial f_y}{\partial y }\right|_{(0,0)} -1\right)y = \lambda_2 y,
\end{equation}
from which we conclude that
\begin{equation}
\lambda_2 > -1.
\label{eqn:lam2}
\end{equation}

\begin{figure}
\begin{center}
\includegraphics[width=0.5\textwidth]{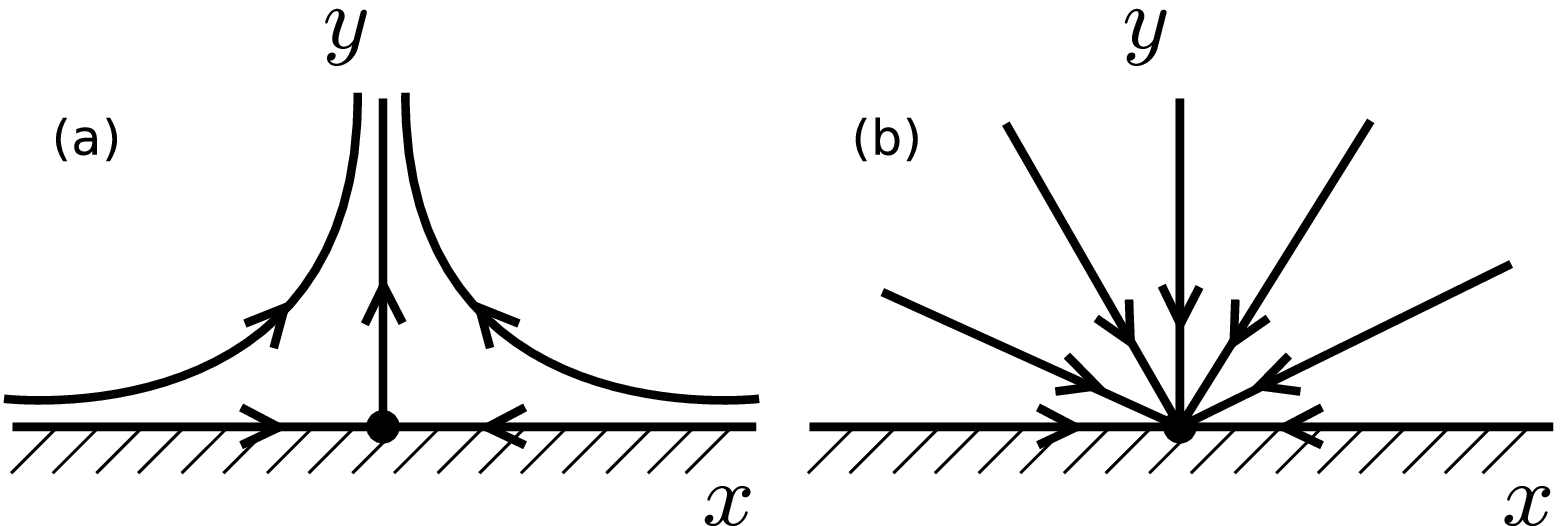}
\caption{Possible local behaviour of a generic boundary fixed point, where $\lambda_1<0$, showing (a) a ``semi-saddle'' point and (b) a ``semi-node''. The latter is possible only if $|Df|\neq 1$ at the fixed point.}
\label{fig:btypes}
\end{center}
\end{figure}

Finally, we use the assumption that $|Df(0,0)| = 1$. Writing $\mu_1$, $\mu_2$ for the eigenvalues of $Df(0,0)$, this gives $\mu_1\mu_2=1$. The eigenvalues of $Df(0,0)$ are related to those of $A$ by $\lambda_1=\mu_1-1$, $\lambda_2=\mu_2-1$, so that
\begin{equation}
\lambda_2 = \frac{-\lambda_1}{1+\lambda_1}.
\end{equation}
If $\lambda_1 > 0$, then clearly $\lambda_2<0$. If $\lambda_1<0$, then we can use \eref{eqn:lam2} to see that $1+\lambda_1>0$, and so $\lambda_2>0$. Thus $\lambda_1$, $\lambda_2$ are real and of opposite sign. This rules out the possibility of a ``semi-node'' on the boundary (\fref{fig:btypes}). There remain two types of generic boundary fixed point: both ``semi-saddles'' but differing in the sign of $\lambda_1$. The two types are shown in \fref{fig:boundaryfp}, along with the form of the extended vector field ${\bf v}_{\tilde{f}}$ in each case. It is evident from the figure that one type has index $0$ while the other has index $-1$. \qed


\end{document}